%% file: local-homophily-transformation.tex
\documentclass[runningheads]{llncs}
\usepackage[T1]{fontenc}
\usepackage{graphicx}
\input{preamble}
\begin{document}
\title{Local Homophily on Bicolored Graphs is $\Poly$-complete}
\author{Pablo Concha-Vega}
\institute{
	Aix Marseille Univ, CNRS, LIS, Marseille, France\\
	\email{pablo.concha-vega@lis-lab.fr}
}

\maketitle              

\input{content}

\bibliographystyle{splncs04}
\bibliography{refs}

\input{appendix}
\end{document}

%% file: preamble.tex
\usepackage{amsmath}
\usepackage{xspace}
\usepackage{amssymb}
\usepackage{tikz}
\usepackage{subcaption}
\usepackage{hyperref}
\graphicspath{{figures/}}

\newcommand{\Poly}{{\mathbf{P}}}

\newcommand{\LS}{{\mathbf{LOGSPACE}}}

\newcommand{\LHE}{\textsf{LHE}\xspace}
\newcommand{\CVP}{\textsf{CVP}\xspace}

\newcommand{\N}{\mathbb{N}\xspace}

\DeclareMathOperator{\level}{\mathsf{level}}

%% file: content.tex
\begin{abstract}
	We propose a local transformation on bicolored graphs, which we call local homophily,
	inspired by adaptive networks and based on majority dynamics and homophily.
	In this transformation, a vertex updates its color to match the majority of its neighbors,
	while neighbors of the same color become connected and neighbors of the opposite color
	become disconnected.

	We show how to simulate Boolean circuits using local homophily
	and establish that determining whether a given pair of vertices becomes connected
	under iterative applications of local homophily is $\Poly$-complete under
	logspace reductions.

	\keywords{Adaptive Networks \and Computational Complexity \and Majority Dynamics \and Homophily}
\end{abstract}

\section{Introduction}

Adaptive networks are systems in which the colors of vertices and the network topology
coevolve over time, with vertex colors and edges updated according to rules that depend
on the current network configuration \cite{gross08,sayama13}.
A central example is consensus formation via majority dynamics, where each vertex updates
its color according to the majority of its neighbors \cite{kozma08,nguyen20,cruciani21}.
Homophily, the tendency of vertices (agents) with similar colors (opinions) to preferentially connect,
often guides the rewiring of edges \cite{mcpherson01,talaga20,loveland25}.

In this work, we study a concrete local transformation on bicolored graphs. Given a vertex 
$v$, the transformation proceeds as follows:
$v$ changes its color to match the majority color of its neighbors; neighbors of 
$v$ that have the same color become connected, while neighbors of the opposite color
become disconnected.

From a computational perspective, a central task in combinatorics is to ``test'' or ``probe''
specific combinations of elements within a combinatorial object. Motivated by this,
we consider the following problem: given a bicolored graph and a pair of special vertices,
does iteratively applying the vertex transformation eventually connect the special pair?

We show that this problem is $\Poly$-complete under logspace reductions via a reduction
from the Circuit Value Problem (CVP). A key structural lemma describing the evolution of the
graph under the transformation underlies the correctness of the reduction and the proof of
the main theorem.
Thus, while adaptive networks are often studied from a dynamical or statistical perspective,
we investigate them from the viewpoint of computational complexity.

\begin{figure}[t]
	\centering
	\includegraphics[width=4cm]{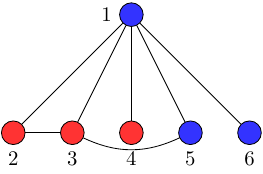}
	\qquad
	\includegraphics[width=4cm]{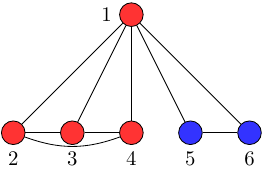}
	\caption{
		Example of local homophily.
		{\bf Left:} the original bicolored graph $B$.
		{\bf Right:} the bicolored graph after applying the transformation at vertex $1$, i.e., $\Phi_1(B)$.
	}\label{fig:example}
\end{figure}

\section{Preliminaries}

We consider finite, simple graphs. Let $G$ be a graph. The set of vertices
of $G$ is denoted $V(G)$, and the set of edges by $E(G) \subseteq V(G) \times V(G)$
($\subseteq {V(G) \choose 2}$ for undirected graphs). When $G$ is undirected,
the neighborhood of a vertex $v$ is denoted by $N_G(v)$.
Given a set $X \in V(G)$, $G[X]$ denotes the subgraph induced by $X$.
For a finite set $X$, we denote by $X^*$ the set of all finite words over $X$,
including the empty word $\varepsilon$. We use the terms \textit{sequence} and
\textit{word} interchangeably. A word $w \in X^*$ of length $n$ is written as
$w = (w_1, w_2, \dots, w_n)$, where $w_i \in X$, or simply as $w_1 w_2 \dots w_n$
when the meaning is clear from the context.
We denote by $[n]$ the set of the first $n$
positive integers, i.e., $[n] = \{1,\dots,n\}$ and $[k,n] = \{k, \dots, n\}$.
A \emph{bicolored graph} is a pair $B=(G, c)$, where $G$ is an undirected graph and
$c: V(G) \to \{-1,+1\}$ is a binary coloring function.

\paragraph{Local Homophily Transformation}
(Fig.~\ref{fig:example}).
Let $B=(G,c)$ be a bicolored graph. For $v \in V(G)$,
the \emph{local homophily transformation} at $v$ is
\[
\Phi_v(B) = (G',c'),
\]
where $c'(u)=c(u)$ for all $u\neq v$, and
\[
c'(v) :=
\begin{cases}
\mathrm{sign}\!\left(\sum_{u \in N(v)} c(u)\right),
& \text{if } \sum_{u \in N(v)} c(u) \neq 0,\\
c(v), & \text{otherwise.}
\end{cases}
\]
The edge set of $G'$ is obtained from $G$ by replacing the
subgraph induced by $N(v)$ with a graph in which
$xy \in E(G')$ if and only if $c(x)=c(y)$ for all distinct
$x,y \in N(v)$. All other edges remain unchanged.

For a word $w = v_1 v_2 \dots v_k \in V(G)^*$, we define
\[
\Phi_w := \Phi_{v_k} \circ \dots \circ \Phi_{v_1}
\quad\text{and}\quad
\Phi_\varepsilon := \mathrm{id}.
\]
Accordingly, $\Phi_w(B)$ denotes the bicolored graph obtained by
applying the local homophily transformation to the vertices of $w$
from left to right.

Given a bicolored graph $B$ and a sequence of vertices $w \in V(G)^*$,
the transformation $\Phi_w(B)$ completely determines the colors of the vertices
and the edges among neighbors affected by the sequence.
A natural computational question is whether two given vertices become adjacent
after applying such a sequence. This leads to the following decision problem:

\paragraph{Local Homophily Evaluation}

(\LHE).
\smallskip

\noindent
\textbf{Input:} A bicolored graph $B=(G,c)$, two vertices $s,t \in V(G)$,
and a word $w \in V(G)^*$.
\smallskip

\noindent
\textbf{Question:} Is $st \in E(G_w)$, where $\Phi_w(B)=(G_w,c_w)$?

\subsection{$\Poly$-completeness Theory}

In order to contextualize the results, we give a brief introduction to
$\Poly$-completeness theory and its main concepts.

\paragraph{Boolean circuits.}

A \emph{Boolean circuit} is a directed acyclic graph whose vertices
(called \emph{gates}) are either inputs, constant values, or logical operations
(such as AND, OR, NOT). A designated vertex is the \emph{output gate}.
Each gate computes a Boolean value from its predecessors according to its operation.
A Boolean circuit is \emph{monotone} if it contains only AND and
OR gates, i.e., it does not have any NOT gates.

\paragraph{Complexity classes.}
A problem is in $\Poly$ if there exists a deterministic Turing machine
that decides it in time polynomial in the size of the input.
A function $f$ is \emph{logspace computable} if there exists a deterministic Turing machine
with a read-only input tape, a write-only output tape, and a read-write work tape,
such that on every input $x$ it outputs $f(x)$ while using at most $O(\log |x|)$ cells on the work tape.
The space used on the input and output tapes is not counted.



\paragraph{Hardness and Completeness.}
Let $\mathbf{C}$ be a class of decision problems, and let $R$ be a type of reduction
(such as logspace reductions).
A problem $A$ is said to be \emph{$\mathbf{C}$-hard} under $R$ if every problem in
$\mathbf{C}$ can be reduced to $A$ using a reduction of type $R$.

If $A$ is both $\mathbf{C}$-hard under $R$ and belongs to $\mathbf{C}$,
then $A$ is called \emph{$\mathbf{C}$-complete} under $R$.

\paragraph{Circuit Value Problem (\CVP).}
The \emph{Circuit Value Problem} is as follows: given a Boolean circuit
and an assignment of its input values, decide whether the output gate
evaluates to TRUE. \CVP is known to be $\Poly$-complete under logspace reductions~\cite{ladner75}.
Moreover, several restricted versions of \CVP remain $\Poly$-complete, including:

\begin{itemize}
	\item the \emph{monotone \CVP}, where only monotone gates (AND and OR) are allowed;
	\item the \emph{synchronous \CVP}, where a circuit is said to be \emph{synchronous} if,
	for every gate $v$, all its input vertices have $\level(u) = \level(v)-1$,
	with $\level(v)$ defined as the length of the longest directed path from an input
	vertex to $v$;
	\item versions with fan-in and fan-out at most $2$.
\end{itemize}

Furthermore, the version of \CVP that combines all of these restrictions
is also $\Poly$-complete~\cite{greenlaw95}.
Although these restricted versions have their own names in the literature,
from now on we will simply refer to \CVP as the version that is monotone,
synchronous, and has fan-in and fan-out at most $2$.

\section{Hardness of \LHE}

In this section, we prove the main result of this article:

\begin{theorem}\label{thm:main}
	\LHE is $\Poly$-complete under logspace reductions.
\end{theorem}

To prove this, we give a logspace reduction from the \CVP.
Intuitively, we encode the evaluation of a Boolean circuit using
the local homophily transformation, so that the output gate evaluates
to TRUE if and only if a specific edge appears in the final graph.
Note that \LHE is trivially in $\Poly$, as it can be solved by sequentially
applying the transformations.

\paragraph{Outline of the proof.}
The proof proceeds along the following steps:
\begin{itemize}
	\item We first prove a structural lemma that will be used in the reduction.
	\item We then present the gadgets corresponding to circuit elements: AND, OR,
		signal duplicator.
	\item Next, we describe how to connect these gadgets to encode an entire
		Boolean circuit, and argue why the sequence of transformations propagates
		the computation correctly.
	\item Finally, we discuss why the construction can be performed in logspace.
\end{itemize}

\subsection{Flower Graphs}

\begin{definition}[Flower graph (Fig.\ref{fig:flower})]
We define a \emph{flower graph} $F_{n,m}$ as follows:
start with a clique $K_n$ (the \emph{center}) with vertices colored $+1$, 
and for each vertex $v \in K_n$, attach a disjoint clique $K_m^{(v)}$ (the \emph{petals}) 
colored $-1$, connecting $v$ to every vertex of $K_m^{(v)}$. No other edges are added.
\end{definition}

\begin{figure}[t]
	\centering
	\includegraphics[width=4cm]{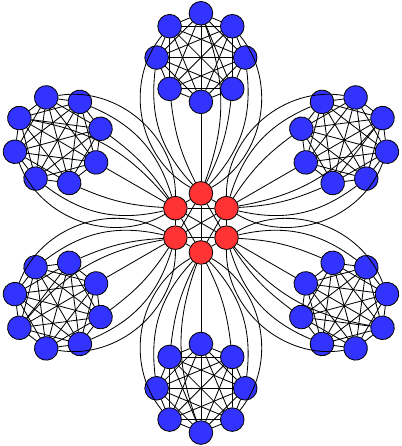}
	\qquad
	\includegraphics[width=4cm]{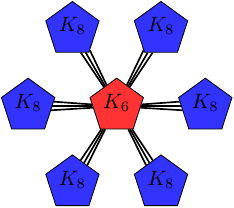}
	\caption{
		Example of a flower graph.
		{\bf Left:} the flower graph $F_{6,8}$.
		{\bf Right:} a succinct representation of $F_{6,8}$.
	}\label{fig:flower}
\end{figure}

\begin{lemma}\label{lemma:flower}
	Let $n,m \in \N$ with $m \ge n$, and let $F_{n,m}$ be a flower graph with central
	vertices $[n]$. For any integer $0 < k \leq n$, the following hold for
	$\Phi_{1 \cdots k}(F_{n,m})$:
	\begin{enumerate}
		\item Each central vertex $i \leq k$ has its color flipped.
		\item For each $i<k$, vertex $i$ is connected to all the vertices of the petal of vertex $i+1$.
		\item The subgraph induced by $[k]$ is the path $P_k$.
		\item For all $i<k$, vertex $i$ has no edges to any central vertex $> k$.
		\item The subgraph induced by $[k,n]$ is $K_{n-k+1}$.
	\end{enumerate}
\end{lemma}
\begin{proof}
	We proceed by induction on $k$.

	Vertex $1$ flips its color when applying $\Phi_1$,
	because in the central clique it has $n-1$ neighbors of its initial color $+1$,
	while in its petal it is connected to $m \ge n$ vertices of the opposite color $-1$.
	Therefore, the majority color in its neighborhood is $-1$, and the flip occurs.

	Conditions 2-5 are easily verified for $k=1$:
	\begin{itemize}
		\item Condition 2 is vacuously satisfied since there are no vertices $i<1$.
		\item Condition 3: the subgraph induced by $\{1\}$ is trivially $P_1$.
		\item Condition 4: there are no vertices $i<1$, so vacuously satisfied.
		\item Condition 5: the subgraph induced by $[1,n]$ is the clique $K_n$ as in the original flower graph.
	\end{itemize}

	{\bf Inductive step.}
	Assume the lemma holds for some $k<n$, i.e., after applying $\Phi_{1\cdots k}$.
	Now consider $\Phi_{1\cdots (k+1)} = \Phi_{k+1} \circ \Phi_{1\cdots k}$.

	\textbf{Color flip:} Vertex $k+1$ flips its color because the majority of its neighbors is $-1$:
	it has $n-k+1$ neighbors in the central clique of color $+1$, one of color $-1$ (vertex $k$)
	and $m \ge n$ neighbors in its petal of color $-1$.

	\textbf{Connections with petals:} By the definition of $\Phi$, vertex $k$ is now connected
	to all vertices of the petal of $k+1$, while the previous connections of vertices $i<k$ remain unchanged.

	\textbf{Path among first $k+1$ vertices:} The subgraph induced by $[k+1]$ forms the path $P_{k+1}$:
	vertices $1$ through $k$ already induce $P_k$ by the inductive hypothesis. Also by the
	inductive hypothesis there is no vertex $<k$ connected to vertices $>k$, which includes $k+1$.
	Therefore, no vertex $i<k$ is connected to $k+1$.

	\textbf{No edges from $[1,k]$ to vertices $>k+1$:}
	The operation $\Phi_{k+1}$ only affects the neighborhood of vertex $k+1$, which consists of $[k,n]$
	together with its petal.
	Since vertex $k$ changed color in the previous step, it becomes disconnected from all vertices $> k+1$,
	while vertices $i<k$ remain disconnected from vertices $> k+1$ by the inductive hypothesis.

	\textbf{Clique among remaining vertices:}
	The subgraph induced by $[k+1,n]$ remains a clique $K_{n-k}$,
	as $\Phi_{k+1}$ does not remove edges among vertices $> k+1$ since they are all of their
	initial color, with the exception of $k+1$. However by definition $\Phi_{k+1}$ does
	not change its neighborhood.\qed
\end{proof}

\subsection{Implementation of Logical Gates via $\Phi$}

In this subsection we construct the gadgets used in the reduction from the \CVP.
We implement logical AND and OR gates, as well as a signal duplicator, within our
bicolored graph framework.
Boolean values are encoded through the presence or absence of specific edges,
and the operation $\Phi$ propagates these values across the construction.

Each gadget consists of a bicolored graph $B$ together with a fixed sequence of
vertices $w$, written as a pair $(B,w)$.
The sequence $w$ specifies the order in which $\Phi$ is applied to the vertices
of the gadget, ensuring that the outputs behave as intended given the inputs.
In other words, $\Phi_w(B)$ evaluates the gadget correctly according to the logical
operation it is meant to simulate.

We start with the OR gadget (see Fig.~\ref{fig:or_and} left),
defined as the pair $(B_\lor, w_\lor)$, where
$B_\lor = (G_\lor, c_\lor)$ is the bicolored graph with
$V(G_\lor) = \{1, \dots, 6\}$,
$E(G_\lor) = \{(1,5), (3,5), (2,6), (4,6)\}$,
and $c_\lor(v) = -1$ for all $v \in V(G_\lor)$.
The fixed sequence of updates is $w_\lor = (1,4,2,3)$.
The inputs are encoded via the potential edges $(1,2)$ and $(3,4)$,
and the output corresponds to the edge $(5,6)$.
The operation of this gadget is illustrated in Appendix A Fig.~\ref{fig:or_functioning}.
Since it is an OR gate, the edge $(5,6)$ is absent if and only if
both edges $(1,2)$ and $(3,4)$ are absent in the initial configuration.

\begin{figure}[t]
	\centering
	\includegraphics[width=3cm]{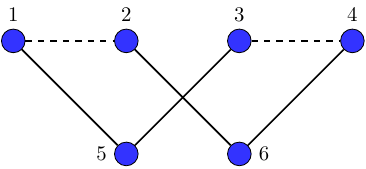}\qquad
	\includegraphics[width=3cm]{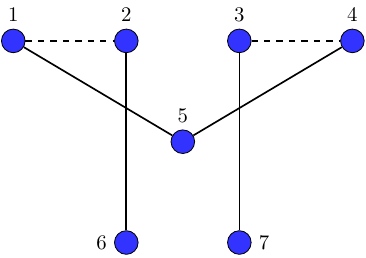}
	\caption{
		OR and AND gadgets. Dashed lines indicate the positions of the inputs.
	}\label{fig:or_and}
\end{figure}

On the other hand, the AND gadget (see Fig.~\ref{fig:or_and} right) is defined
as the pair $(B_\wedge, w_\wedge)$, where
$B_\wedge = (G_\wedge, c_\wedge)$ is the bicolored graph with
$V(G_\wedge) = \{1,\dots,7\}$,
$E(G_\wedge) = \{(1,5), (2,6), (3,7), (4,5)\}$,
and $c_\wedge(v) = -1$ for all $v \in V(G_\wedge)$.
The fixed sequence of updates is $w_\wedge = (2,3,1,4,5)$.
The inputs are encoded via the potential edges $(1,2)$ and $(3,4)$,
and the output corresponds to the edge $(6,7)$.
The operation of this gadget is illustrated in Appendix A Fig.~\ref{fig:and_functioning}.
Since it is an AND gate, the edge $(6,7)$ is present if and only if
both edges $(1,2)$ and $(3,4)$ are present in the initial configuration.

Although we now have OR and AND gadgets, for our reduction we require all
gadgets to have outdegree at most $2$.
As is clear, our current gadgets have a single output,
so we need to introduce an additional gadget capable of duplicating a Boolean value.  
This means the gadget has one input and two outputs, each of which will be
TRUE if and only if the input was initially TRUE.

The duplicator gadget is the most complex gadget in our construction,
at least among the ones we could find.
It consists of several cliques and, unlike the logical gadgets, it has
an additional parameter $k$, which we will discuss later.

Concretely, the gadget contains vertices labeled $1$ through $11$,
with vertex $1$ connected to the center of a flower graph $F_{k,k+2}$.

The duplicator gadget (see Fig.~\ref{fig:dupl}) is defined as the pair $(B_D, w_D)$, where
$B_D = (G_D, c_D)$ is the bicolored graph with
$V(G_D) = \{1,\dots,11\} \cup V(F_{k,k+2})$,
and vertex $1$ is connected to the center of the flower graph $F_{k,k+2}$,
which we assume are labeled from $12$ to $11+k$.
The edges outside the flower are given by:
\[
\begin{aligned}
E(G_D) \setminus E(F_{k,k+2}) = \{ & (1,3), (1,4), (1,6), (2,5), (2,7), \\
& (3,6), (4,8), (4,9), (4,10), (5,11) \}.
\end{aligned}
\]
The coloring of the vertices is defined as
\[
c_D(v) =
\begin{cases}
+1 & \text{for } v = 4,5 \text{ and } v \in [12,11+k], \\
-1 & \text{for all other vertices.}
\end{cases}
\]
The fixed update sequence is $w_D = (1, 2, 12, \dots, 11+k, 3, 1, 4, 5, 4)$.
The single input is encoded via the potential edge $(1,2)$,
and the two outputs correspond to the edges $(6,7)$ and $(10,11)$.
The operation of this gadget is depicted in 
Fig.~\ref{fig:dupl_true_functioning} and Fig.~\ref{fig:dupl_false_functioning}.
Both outputs are TRUE if and only if the input edge is initially present.

Note that upon applying $\Phi_1$, the center of the flower becomes connected to vertex $4$.  
Subsequently, after applying $\Phi_2$, this connection is maintained.
Since vertices $1$ and $4$ are both connected to the center of the flower,
this structure can be interpreted as a flower $F_{k+2,k+2}$.  
Therefore, after applying $\Phi_{(12,\dots,11+k)}$, we obtain exactly the configuration
described in Lemma~\ref{lemma:flower}; that is, the center of the flower has reversed its color
and only vertex $11+k$ is connected to $1$ and $4$.
The remaining transformations can be carried out without difficulty.

The duplicator gadget also satisfies several important properties that are crucial for our construction:  
\begin{enumerate}
    \item Vertices belonging to different outputs are never connected.  
    \item Vertices of the flower are never connected to the outputs.  
    \item The output vertices are all colored $-1$.
\end{enumerate}

\begin{figure}[t]
	\centering
	\includegraphics[width=8cm]{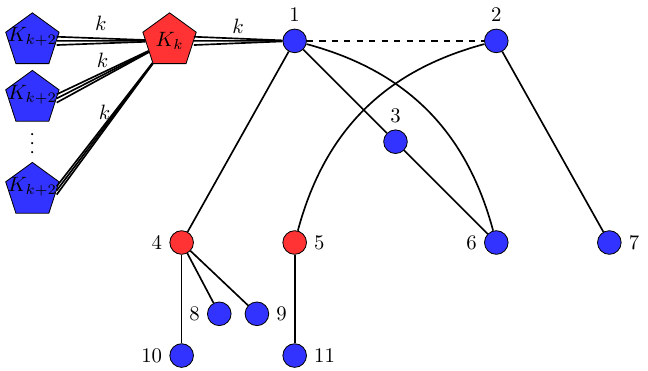}
	\caption{
		Duplicator gadget.
	}\label{fig:dupl}
\end{figure}

\begin{figure}[]
	\centering
	\newcounter{subfig}
	\foreach \i in {0,...,8}{
		\stepcounter{subfig}
		\begin{subfigure}[b]{5.5cm}
			\includegraphics[width=\textwidth]{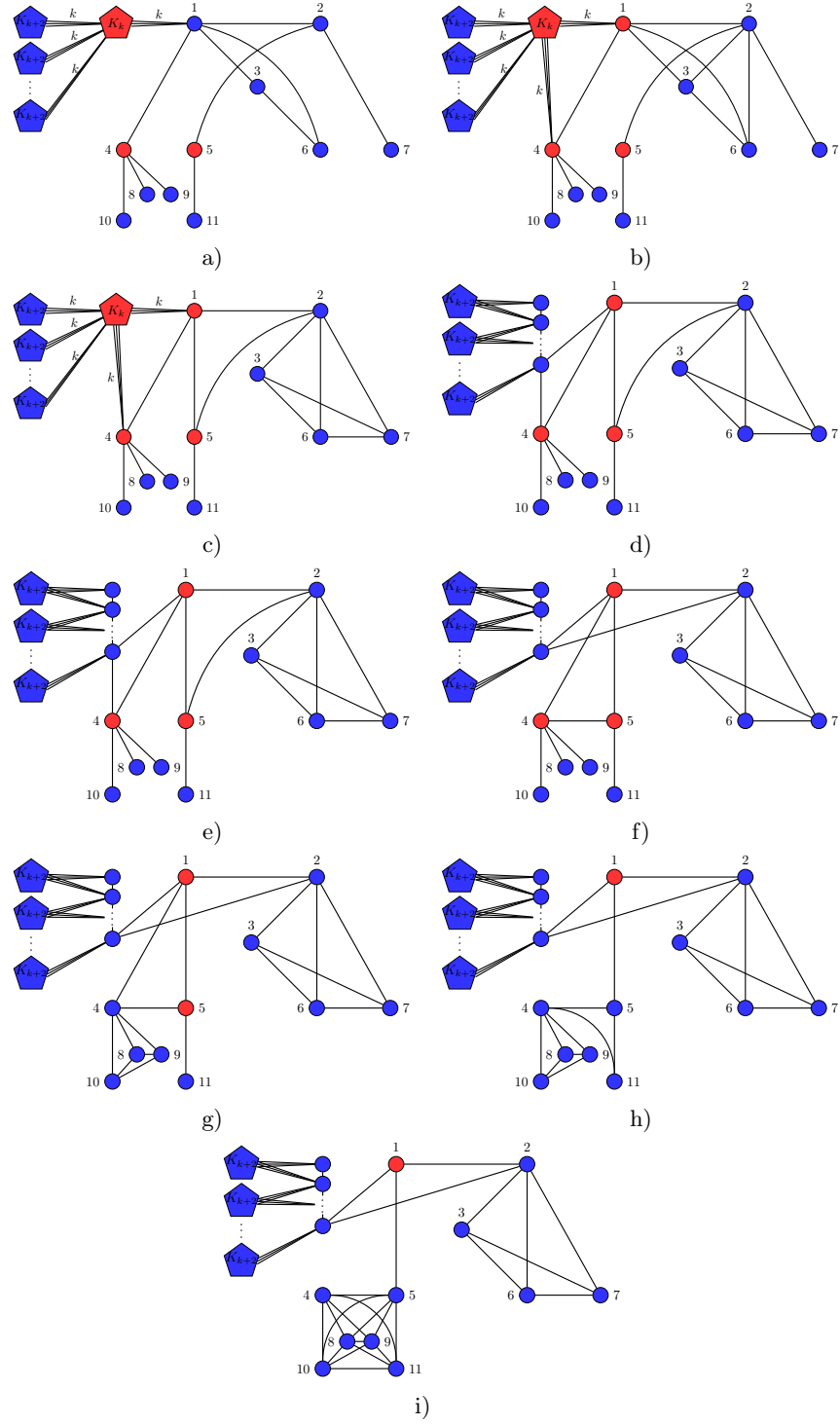}
			\caption*{\alph{subfig})}
		\end{subfigure}
    }
    \caption{
		{\bf a)} The duplicator with input set to TRUE.
		{\bf b)--i)} Evolution of the gadget under $\Phi_{w_D}$.
		The transition from {\bf c)} to {\bf d)} follows from Lemma~\ref{lemma:flower}.
    }\label{fig:dupl_true_functioning}
\end{figure}

\begin{figure}[]
	\centering
	\setcounter{subfig}{0}
	\foreach \i in {0,...,8}{
		\stepcounter{subfig}
		\begin{subfigure}[b]{5.5cm}
			\includegraphics[width=\textwidth]{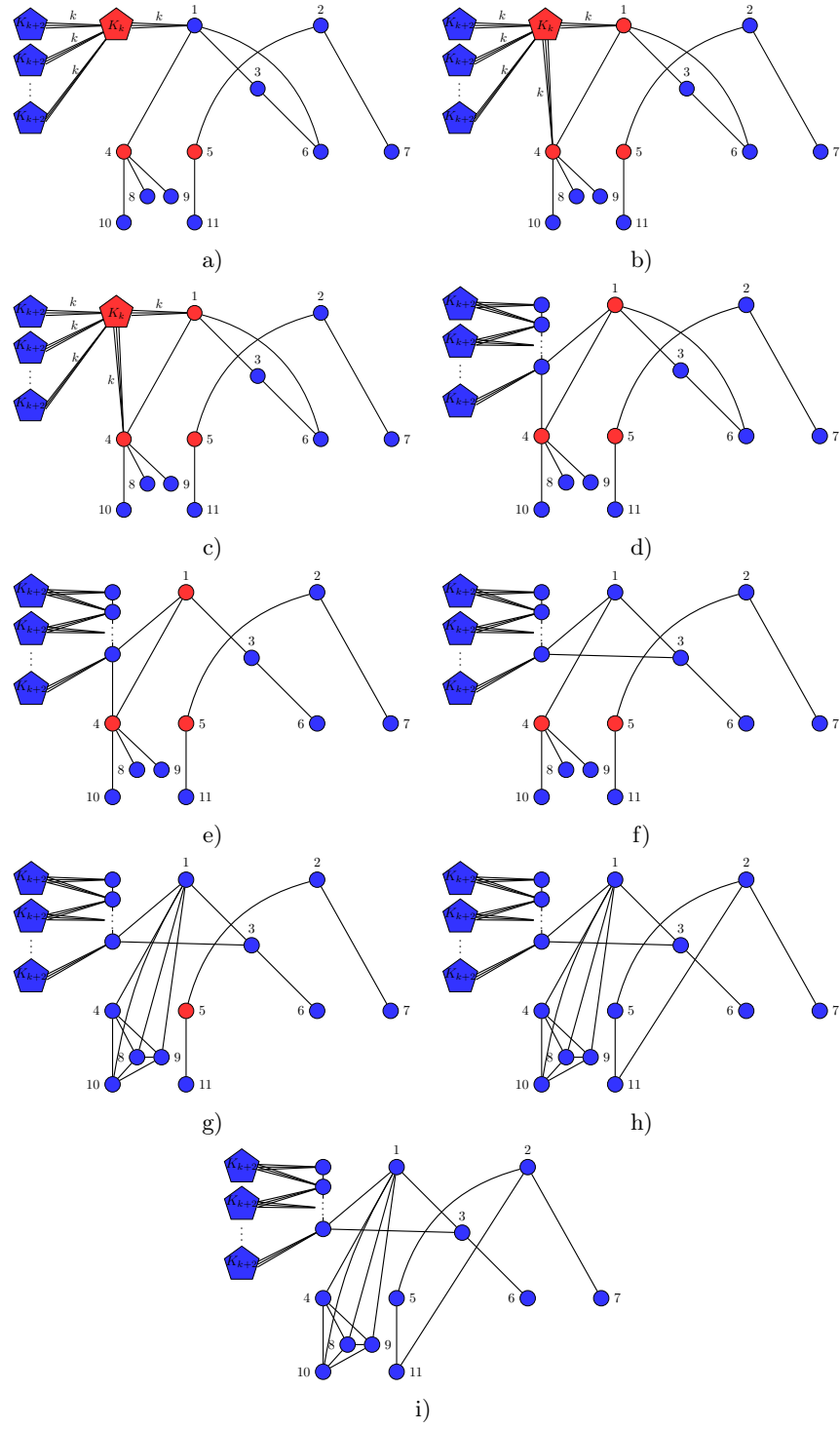}
			\caption*{\alph{subfig})}
		\end{subfigure}
    }
    \caption{
		{\bf a)} The duplicator with input set to FALSE.
		{\bf b)--i)} Evolution of the gadget under $\Phi_{w_D}$.
		The transition from {\bf c)} to {\bf d)} follows from Lemma~\ref{lemma:flower}.
    }\label{fig:dupl_false_functioning}
\end{figure}

\newpage
\subsection{Simulating Boolean Circuits}

Given a monotone Boolean circuit, we construct a corresponding network of gadgets as follows.  
For each circuit input, we introduce a duplicator gadget.  
If an input is TRUE, the corresponding input edge of its duplicator is added to the initial configuration.

For each AND and OR gate, we introduce the corresponding logical gadget and connect it to a duplicator.  
Connecting two gadgets means identifying the pair of output vertices of one gadget 
with the predefined input vertices of the next.  

Let the level of a gate be the length of the shortest path from an input to that gate.  
The construction is organized in layers according to the level: 
a layer of duplicators (inputs), followed by a layer of gates, then duplicators, and so on.  
The global update sequence is obtained by concatenating the fixed sequences of all gadgets 
in increasing order of their level.

A key issue is the propagation of edges across layers.  
After a gadget is evaluated, its output vertices may acquire new edges.  
Since these vertices serve as inputs to gadgets in the next layer, such edges accumulate 
towards lower levels.

This is where the parameter $k$ of the duplicator becomes crucial.  
Vertex $1$ is connected to the center of a flower of size $k$.  
To ensure that $\Phi_1$ flips the color of vertex $1$, 
the number of $-1$ colored neighbors potentially attached to it from upper layers 
must be strictly smaller than $k$.

An upper bound on this number is
\[
7 n_\wedge + 6 n_\lor + 9 (n_\wedge + n_\lor),
\]
where $n_\wedge$ and $n_\lor$ denote the number of AND and OR gates, respectively.  
Choosing $k$ strictly larger than this quantity guarantees that the accumulated influence 
from upper layers cannot prevent the intended color change of vertex $1$.

\begin{theorem}
Let $C$ be a monotone Boolean circuit and let $(B_C, w_C)$
be the bicolored graph and update sequence obtained by the above construction.
For every assignment to the input gates of $C$,
the designated output edge of $B_C$ is present after applying $w_C$
if and only if the output gate of $C$ evaluates to TRUE.
\end{theorem}

\begin{proof}
We prove the statement by induction on the level of the circuit.

For level $0$ (input gates), each input gate is represented by a duplicator gadget.
By construction, the presence or absence of its designated input edge
encodes the Boolean value assigned to that gate.
By correctness of the duplicator gadget,
its output edges are present if and only if the input edge is present.
Hence the duplicator correctly encodes the value of each input gate.

Assume now that for every gate at level at most $\ell$,
the corresponding gadget in $B_C$ produces an output edge
that is present if and only if the gate evaluates to TRUE.
Let $g$ be a gate at level $\ell+1$.
By construction, the input vertices of the gadget corresponding to $g$
are identified with the output vertices of gadgets at level $\ell$.
By the induction hypothesis,
these vertices encode exactly the Boolean values
of the predecessors of $g$ in the circuit.

Since the AND and OR gadgets were shown to compute
the correct Boolean function on their designated input edges,
the output edge of the gadget corresponding to $g$
is present after applying its update sequence
if and only if the gate $g$ evaluates to TRUE.

It remains to verify that propagated edges from upper levels
do not alter the intended behavior of duplicators.
By the choice of $k$, the number of additional $-1$ colored neighbors
that may accumulate at any duplicator
is strictly smaller than the size of the flower center.
Therefore, the application of $\Phi_1$
produces the same color change as in the isolated analysis of the duplicator,
and its outputs depend only on its designated input.

Finally, since $w_C$ is defined as the concatenation of the update sequences
of all gadgets in increasing order of their level,
every gadget is evaluated only after all gadgets at smaller levels
have been evaluated.
Hence the above argument applies to all levels.

Therefore, the designated output edge of $B_C$
is present after applying $w_C$
if and only if the output gate of $C$ evaluates to TRUE.\qed
\end{proof}

\subsection{The Construction is Computable in $\LS$}

The parameter $k$ can be computed in logarithmic space by counting the number of AND and OR gates and applying
\(
k > 7 n_\wedge + 6 n_\lor + 9(n_\wedge + n_\lor),
\)
using $O(\log n)$ bits for the counters.

The vertex set of $G_C$ consists of gadgets associated with each gate, as well as duplicators gadgets.
Although the duplicator gadget depends on a parameter $k = O(n)$, each vertex can be identified by a pair (gate index, local gadget index), which requires
$O(\log n)$ space.
Edges are generated on-the-fly by iterating over gates and their connections in the input circuit.
For each output edge, the corresponding input vertex of the next gadget is reconstructed
from its gate index and local index, so there is no need to store entire layers of outputs.

The coloring $c_C$ is determined locally by gadget type and local index,
and is written simultaneously when generating vertices.
The update sequence $w_C$ is generated by iterating over gates in topological order and writing
each gadget's fixed sequence in order.
Since the construction of each gadget only requires $O(\log n)$ space for indices and counters,
the entire construction $(B_C, w_C)$ can be computed in logarithmic space.
This concludes Theorem~\ref{thm:main}.

\section{Conclusions}
In this work, we introduced a local graph transformation, which we call local homophily,
inspired by adaptive networks and based on the principles of majority dynamics and homophily.
We defined the class of flower graphs and proved a key structural lemma that underlies the
functioning of the duplication gadget. We then presented the construction of the gadgets,
explained how to assemble them into circuits, and proved the correctness of the overall construction.
Finally, we discussed how the construction can be computed in logspace, leading to our main
theorem establishing the $\Poly$-completeness of the \LHE.

Several directions for future research emerge from this work. One natural question is to
further study the dynamics of the local homophily transformation itself, for example by
characterizing the set of graphs reachable from a given class of initial graphs.
Potential applications could include consensus or other complex processes on adaptive networks.
It would also be interesting to analyze the computational complexity of other local transformations,
or to consider complexity questions on local homophily itself, such as reachability:
``given two bicolored graphs, can one be reached from the other via local homophily?''
Finally, one could explore richer algebraic structures associated with the transformation,
which may provide additional insights into its behavior and computational properties.

%% file: appendix.tex
\appendix
\section{Full Behavior of the Gates}

This appendix provides figures illustrating the full behavior of the
OR and AND gadgets.

\begin{figure}[t]
	\centering
	\setcounter{subfig}{0}
	\foreach \i in {0,...,4}{
		\stepcounter{subfig}
		\begin{subfigure}[b]{3cm}
			\includegraphics[width=\textwidth]{gadget_or_true_false_\i.pdf}
			\caption*{\alph{subfig})}
		\end{subfigure}
    }\\
	\foreach \i in {0,...,4}{
		\stepcounter{subfig}
		\begin{subfigure}[b]{3cm}
			\includegraphics[width=\textwidth]{gadget_or_true_true_\i.pdf}
			\caption*{\alph{subfig})}
		\end{subfigure}
    }
    \caption{
		Behavior of the OR gadget:  
		{\bf a)} OR gadget with one input TRUE and one input FALSE.  
		{\bf b)--e)} Evolution according to $w_\lor$.  
		{\bf f)} OR gadget with both inputs TRUE.  
		{\bf g)--j)} Evolution according to $w_\lor$.  
		The case with both inputs FALSE is trivial.  
		The case with inputs FALSE and TRUE is symmetrical to TRUE and FALSE.
    }\label{fig:or_functioning}
\end{figure}

\begin{figure}[]
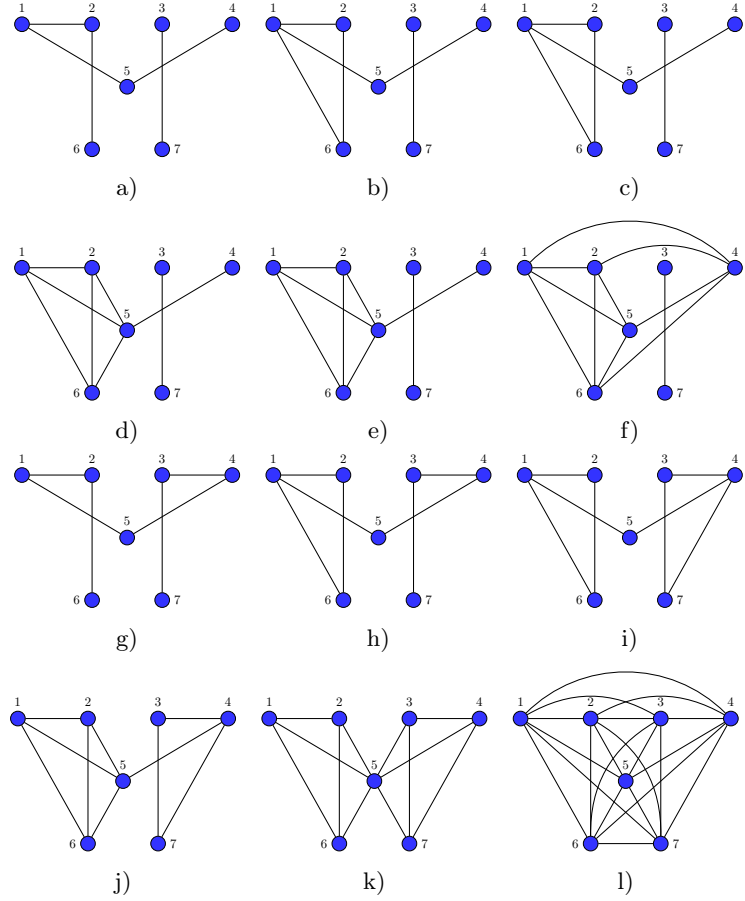

	\centering
	\setcounter{subfig}{0}
	\foreach \i in {0,...,5}{
		\stepcounter{subfig}
		\begin{subfigure}[b]{3cm}
			\includegraphics[width=\textwidth]{gadget_and_true_false_\i.pdf}
			\caption*{\alph{subfig})}
		\end{subfigure}
    }\\
	\foreach \i in {0,...,5}{
		\stepcounter{subfig}
		\begin{subfigure}[b]{3cm}
			\includegraphics[width=\textwidth]{gadget_and_true_true_\i.pdf}
			\caption*{\alph{subfig})}
		\end{subfigure}
    }
    \caption{
		Behavior of the AND gadget:  
		{\bf a)} AND gadget with one input TRUE and one input FALSE.  
		{\bf b)--f)} Evolution according to $w_\wedge$.  
		{\bf g)} AND gadget with both inputs TRUE.  
		{\bf h)--l)} Evolution according to $w_\wedge$.  
		The case with both inputs FALSE is trivial.  
		The case with inputs FALSE and TRUE is symmetrical to TRUE and FALSE.
    }\label{fig:and_functioning}
\end{figure}